\documentclass[a4paper,UKenglish]{lipics-v2021}

\usepackage{mathtools}
\usepackage{booktabs}
\usepackage{float}
\usepackage{tikz}
\usetikzlibrary{arrows.meta,fit,backgrounds}

\pdfoutput=1
\hideLIPIcs
\nolinenumbers

\title{Two-Variable Logic with Arbitrarily Many Successor Relations}
\titlerunning{Two-Variable Logic with Arbitrarily Many Successor Relations}

\author{Jakub Michaliszyn}
{University of Wrocław, Wrocław, Poland}
{}
{}
{}

\author{Piotr Witkowski}
{University of Wrocław, Wrocław, Poland}
{}
{}
{}

\authorrunning{J. Michaliszyn and P. Witkowski}

\Copyright{Jakub Michaliszyn and Piotr Witkowski}


\newcommand{\FO}{\mathrm{FO}^{2}}
\newcommand{\Rays}{\operatorname{Rays}}
\newcommand{\Match}{\operatorname{Match}}
\newcommand{\Inv}{\operatorname{inv}}
\newcommand{\tp}{\operatorname{tp}}
\newcommand{\NN}{\mathbb{N}}
\newcommand{\ZZ}{\mathbb{Z}}
\newcommand{\QQ}{\mathbb{Q}}
\newcommand{\NEXP}{\mathsf{NEXPTIME}}
\newcommand{\TwoNEXP}{\mathsf{2\text{-}NEXPTIME}}
\newcommand{\Dis}{\operatorname{D}}
\newcommand{\Adj}{\operatorname{Adj}}
\begin{document}
\maketitle
\begin{abstract}
We study two-variable first-order logic with a finite, input-dependent
number of distinguished successor relations and arbitrary additional unary
and binary predicates. Each distinguished relation must be the exact
immediate-successor relation of some linear order on the common domain;
the orders themselves are not available in the language and may have
arbitrary order types. We prove that satisfiability belongs to
$\TwoNEXP$, uniformly in the number of successors.

The proof gives finite certificates for possibly infinite models.
A closure operation separates local neighbourhood descriptions, called
star types, into those of bounded multiplicity and those that can be
realized infinitely often. The bounded part is
represented explicitly. Outside it, integer height vectors allow
overlapping successor requirements to be matched without creating cycles
or unintended identifications. Ordering the resulting path components
densely then recovers exact successor relations. Consistent assignments of the additional binary predicates handle
witnesses involving the bounded part.

The same certificates decide infinite satisfiability in $\TwoNEXP$ and
give a doubly exponential threshold above which a finite model guarantees
an infinite model. A doubly exponentially large finite core can be forced with only two
successors and unary predicates. For finite models, we prove that
satisfiability subject to a bound on the total number of reference
adjacencies lost by the other orders is $\NEXP$-complete, even when
the budget is encoded in binary and the number of successors is part
of the input.
\end{abstract}
\clearpage

\section{Introduction}
We prove that satisfiability for two-variable first-order logic remains
decidable in the presence of \emph{arbitrarily many independent successor
relations}. The number of successors is part of the input, additional
unary and binary predicates are unrestricted, and our
$\TwoNEXP$ upper bound is uniform in all these parameters. Each
distinguished relation is the exact immediate-successor relation of a
linear order on the common domain. The orders may have arbitrary order
types and are absent from the signature.

This result identifies a boundary between access to order and access to
adjacency. Ordinary two-variable first-order logic, $\FO$, uses only two
variable names, reused under quantification; its satisfiability problem
is $\NEXP$-complete and it has the finite-model
property~\cite{mortimer1975,gradel1997}. Requiring three predicates to
be linear orders already makes satisfiability undecidable, even with
only additional unary predicates~\cite{kieronski2011orders}. For
successors, finite satisfiability with two distinguished relations is
$\NEXP$-complete~\cite{figueira2013}, also with additional binary
predicates~\cite{charatonik2016}. Our main theorem allows any
input-dependent number of successors in arbitrary models, with no
hierarchy or restriction on how their edges overlap. Different
successors may share an edge or traverse the same pair in opposite
directions. Here $\NEXP$ and $\TwoNEXP$ denote nondeterministic singly
and doubly exponential time, respectively.

Passing to arbitrary models changes the problem substantially. Every
successor relation on a finite domain is a single directed spanning
path. On an infinite domain, an exact successor graph may have many
components, or even no edges. Moreover, successor constraints destroy
the finite-model property: $\forall x\exists y\,S_1(x,y)$ has a model
over $\ZZ$ and no finite model. An empty successor relation on at least
two elements is realizable over $\QQ$, but cannot occur in a finite
admissible model. A decision procedure must therefore account for
models that are necessarily infinite.

\begin{theorem}\label{thm:main}
Satisfiability of $\FO$ with any finite number of distinguished exact
successor relations belongs to $\TwoNEXP$, uniformly in the number of
successors, with arbitrary additional unary and binary predicates.
It is $\NEXP$-hard already with one distinguished successor.
\end{theorem}

\paragraph*{Finite certificates for infinite models.}
The main technical contribution is a finite certificate that describes
all successor neighbourhoods and consistently supplies the witnesses
required by the formula. Two obstacles have to be handled jointly.
Successor positions belonging to different relations may be forced to
name the same neighbour, so matching colours independently can destroy
local requirements. Some neighbourhoods can occur only finitely often,
so duplicating every local description cannot supply an infinite model.
Additional binary predicates make this finite part especially
delicate: several witness obligations can impose incompatible
assignments on the same pair.

Our construction resolves these obstacles in four steps. A
\emph{star type} records a complete successor neighbourhood, including
shared neighbours. A closure operation propagates multiplicity bounds
through these descriptions and separates the bounded part from types
that can be realized infinitely often. We store the bounded elements
and their neighbours in a finite \emph{interface}; rows record
consistent binary assignments to the bounded elements. Outside the
interface, integer height vectors allow whole neighbourhood masks to
be matched simultaneously, preserving shared edges and excluding
directed cycles in every colour. Finally, a dense ordering of the
resulting path components makes each graph an exact successor
relation. The freedom to choose the invisible underlying orders is
what permits this last step.

\paragraph*{Consequences and the size of the finite part.}
The certificates also decide \emph{infinite satisfiability} in
$\TwoNEXP$ and yield countably infinite models. They give an effective
doubly exponential threshold above which a finite model guarantees an
infinite model (Section~\ref{sec:consequences}). The finite part can
itself be doubly exponential: for every $r\geq1$, a sentence of size
polynomial in $r$, using only two successors and unary predicates,
forces exactly $2^{2^r}-1$ elements in a designated finite set, while
requiring the whole model to be infinite
(Appendix~\ref{sec:large-core}). This establishes a structural
limitation of certificates that list the core explicitly. The exact
complexity between the $\NEXP$ lower bound and our $\TwoNEXP$
upper bound remains open.

\paragraph*{A finite-model counterpart with a binary budget.}
Our second result allows arbitrarily many successors on finite domains
when their adjacency disagreement is bounded. Choose one order as a
reference and count its directed successor edges missing from the
other orders, with multiplicity across orders. Given a budget $b$,
the resulting finite satisfiability problem is $\NEXP$-complete
\emph{even when both the number of successors and the binary-encoded
budget are part of the input} (Section~\ref{sec:budget}). The
restriction places no bound on the domain size. Binary ranks using
$O(\log(b+2))$ unary bits per relation exclude cycles, allowing a
polynomial-size reduction to two-variable counting logic with one
successor. Thus the time bound depends on the length of the budget's
encoding. A separate decomposition into common blocks 
gives a model-size bound singly exponential in the value of $b$.
These results leave unrestricted finite satisfiability with three or
more successors unresolved.

\subsection{Related work}
The closest results concern finite structures with two successors.
The other comparisons below clarify the roles of arbitrary order
types, invisible orders, and unrestricted binary predicates.

\paragraph*{Successors and forests.}
The finite two-successor problem was studied by
Manuel~\cite{manuel2010successors}; Figueira proved
$\NEXP$-completeness with additional unary
predicates~\cite{figueira2013}. Charatonik and Witkowski obtained
finite $\NEXP$-completeness for counting logic over two forests of
ranked trees, including additional binary
predicates~\cite{charatonik2016}. Charatonik, Guskov, Pratt-Hartmann,
and Witkowski lifted the rank restriction and also proved general
satisfiability decidable for one forest~\cite{charatonikEtAl2018}.
Their finite one-forest result supplies the one-successor procedure
used in our binary-budget reduction. Our unrestricted theorem concerns
arbitrary models with an input-dependent number of exact successors;
our finite theorem controls their disagreement.

\paragraph*{Visible orders and arbitrary order types.}
Schwentick and Zeume proved finite-model decidability for two linear
orders and related preorder signatures with additional unary
predicates~\cite{schwentick2012}. Manuel and Zeume treated two
successors with one corresponding order and additional unary
predicates~\cite{manuel2013}; Zeume
and Harwath extended that result to additional binary predicates in
their study of order-invariance~\cite{zeume2016}. Their exponential
small-model bound for one order and its successor also underlies our
finite-model size bound.
For labelled linear orders without a successor predicate, Manuel and
Sreejith obtained $\NEXP$-completeness over arbitrary, countable,
and scattered orders~\cite{manuel2016countable}. Toru\'nczyk and
Zeume proved general satisfiability decidable for two linear orders
with arbitrary additional predicates~\cite{torunczykZeume2022}.
Our use of arbitrary order types has a different purpose: it lets us
order successor components densely after their edges have been fixed.

\paragraph*{Words and data.}
Etessami, Vardi, and Wilke characterized $\FO$ with order and
successor on finite and one-sided infinite words by temporal logic
with unary future and past operators~\cite{etessamiEtAl2002}.
Boja\'nczyk, David, Muscholl, Schwentick, and Segoufin proved
decidability when data equality is also
available~\cite{bojanczykEtAl2011datawords}. These settings combine
one sequence with further structure; we allow many independent
adjacency relations on a common domain.

\paragraph*{Other semantic constraints.}
A single equivalence relation preserves $\NEXP$-completeness and
the finite-model property~\cite{kieronskiOtto2012}. Kiero\'nski,
Michaliszyn, Pratt-Hartmann, and Tendera proved
$\TwoNEXP$-completeness for two equivalence closures, with hardness
already for two equivalence relations; three equivalence relations
give undecidability~\cite{kieronskiEtAl2014}. Fiuk, Kiero\'nski,
and Michielini obtained finite $\NEXP$-completeness for nested
equivalences with a linear order and for nested total preorders;
adding successors of the preorders yields
$\mathsf{EXPSPACE}$-completeness~\cite{fiukEtAl2025}. Our successors
need no nesting.
For transitivity and related constraints, decidability results depend
on finiteness~\cite{prattHartmann2018}, guarded
witnesses~\cite{szwastTendera2019,kieronskiTendera2018}, a universal
fragment with constants and transitive
closure~\cite{kieronskiMichaliszyn2012}, or restricted universal
comparability~\cite{marzecTendera2026}. Our existential witnesses
may lie outside every successor neighbourhood.

\paragraph*{Types and counting.}
Our proof builds on local types and finite exceptional
parts~\cite{kieronskiOtto2012}, and on separating witnesses into
groups~\cite{otto2001,zeume2016}. The additional ingredients are
multiplicity propagation through complete neighbourhoods, consistent
rows to the bounded part, and simultaneous matching at integer heights.
Pratt-Hartmann's $\NEXP$ procedure for two-variable counting logic,
$\mathrm C^2$, uses star types and integer multiplicities, even with
binary-encoded counts~\cite{prattHartmann2005counting}. In counting
logic, one order and its successor yield finite
$\NEXP$-completeness with additional unary predicates; with
additional binary predicates the problem is equivalent under
elementary reductions to multicounter-automata
emptiness~\cite{charatonikWitkowski2016order}. Our budget reduction
uses one successor without an order predicate.

\paragraph*{Reading guide.}
Section~\ref{sec:preliminaries} fixes the semantics and local descriptions.
Sections~\ref{sec:star-types}--\ref{sec:realization} develop the
certificates; Section~\ref{sec:decision} derives the decision procedures.
The finite-model argument in Section~\ref{sec:budget} is independent
and can be read after the preliminaries. The appendices give
normalization details, the large-core construction and a worked
certificate.

\newpage
\section{Preliminaries}\label{sec:preliminaries}
All structures have nonempty domains and finite relational signatures,
with equality, unary and binary predicates, and no constants or
functions. The fragment $\FO$ uses only the variables $x,y$, which may
be reused under quantification. The explicitly represented signature
contains distinguished predicates $S_1,\ldots,S_k$, with $k\geq1$.
All other predicates are unrestricted.

\begin{definition}\label{def:admissible}
A structure on $A$ is \emph{admissible} if for each $i$ there is a
strict linear order $<_i$ on $A$ such that
\[
 S_i(a,b)\quad\Longleftrightarrow\quad
 a<_i b\ \land\ \neg\exists c\,(a<_i c<_i b).
\]
The orders may have arbitrary order types and are not part of the
signature. Satisfiability means existence of an admissible model.
\end{definition}

Write $n$ for the input length, including its signature.
Scott normalization, recalled in Appendix~\ref{app:normalization},
produces in polynomial time a sentence in \emph{Scott normal form}
\begin{equation}
 \Phi=\forall x\forall y\,\alpha(x,y)\ \land\
       \bigwedge_{j=1}^{m}\forall x\exists y\,\beta_j(x,y),
 \label{eq:scott}
\end{equation}
where all matrices are quantifier-free. An original structure satisfies
the input precisely when some expansion by fresh unary predicates
satisfies $\Phi$. Thus normalization preserves the domain and every
original predicate. We assume that $\alpha$ is symmetric, replacing it
by $\alpha(x,y)\land\alpha(y,x)$ if necessary. We call the requirement
to satisfy $\beta_j$ at an element its $j$th \emph{witness obligation}.

\paragraph*{Neighbourhoods and masks.}
The graph of $S_i$ is called the \emph{colour-$i$ graph}.
Two distinct elements $a,b$ are \emph{neighbours} if $S_i(a,b)$ or
$S_i(b,a)$ holds for some $i$. The symbols $i^+$ and $i^-$ denote
the successor and predecessor \emph{ports} of colour $i$: the positions
at which an element can have its $S_i$-successor and $S_i$-predecessor.
Each port has at most one neighbour, giving at most $d=2k$ neighbours
per element.
For distinct $a,b$, their \emph{mask} is
$M(a,b)=\{i^+:S_i(a,b)\}\cup\{i^-:S_i(b,a)\}$.
A legal mask contains at most one sign of each colour, and $M^{-1}$
reverses all signs. A \emph{far pair} is a pair of distinct elements
with empty mask; we also say that its endpoints are \emph{far} from
each other. A nonempty mask records a single neighbour, possibly
shared by several colours: $\{1^+,2^-\}$ means that this neighbour is
both the $S_1$-successor and the $S_2$-predecessor of the source.

\paragraph*{Types and assignments.}
After normalization, let $u$ count the unary predicates and $a$ the
additional binary predicates. A \emph{1-type} records their $u$ unary
values and $a$ diagonal values $R(x,x)$; all $S_i(x,x)$ are false.
Thus it uses $u+a$ bits, giving at most $2^{u+a}$ 1-types.
The parameters $u,a,k,m$ are polynomially bounded in $n$, since
normalization has polynomial size.
Let $\mathcal R$ be the set of additional binary relation symbols, and
let $\mathcal R^{\pm}=\{R,R^{-1}:R\in\mathcal R\}$, where $R^{-1}$ is
a formal converse symbol. A \emph{binary assignment} is a set
$\gamma\in\Gamma=\mathcal P(\mathcal R^{\pm})$. On distinct endpoints
$a,b$, membership of $R$ means $R(a,b)$, and membership of $R^{-1}$
means $R(b,a)$; absent symbols mean false atoms. The inverse
$\Inv(\gamma)$ exchanges $R$ and $R^{-1}$. Each of the $a$ predicates
has two directed truth values, so an assignment uses $2a$ bits and
$|\Gamma|=2^{2a}$.
For 1-types $\sigma,\tau$, a legal mask $M$, and $\gamma\in\Gamma$,
the tuple $(\sigma,\tau,M,\gamma)$ determines a \emph{2-type}: the
complete atomic description of two distinct endpoints, including their
1-types and the literal $x\ne y$.  For a quantifier-free formula
$\theta$, let $\theta[\sigma;=]$ be its truth value when $x,y$ denote
the same element of 1-type $\sigma$, and let
$\theta[\sigma,\tau,M,\gamma]$ be its truth value on that 2-type. The
permitted assignments for a far pair are
\begin{equation}
 A_{\sigma\tau}=\{\gamma\in\Gamma:
       \alpha[\sigma,\tau,\varnothing,\gamma]\}.
 \label{eq:far-table}
\end{equation}
Symmetry gives $\gamma\in A_{\sigma\tau}$ if and only if
$\Inv(\gamma)\in A_{\tau\sigma}$. If there are no additional binary
predicates, $\Gamma$ contains just the empty assignment.

\paragraph*{Paths and exact successors.}
A directed graph is \emph{functional} if each vertex has at most one
outgoing edge, and \emph{inverse-functional} if each has at most one
incoming edge. Throughout, \emph{acyclic} means having no directed
cycle, and components are taken after forgetting edge directions.
A functional, inverse-functional, acyclic graph consists of paths with
order types finite, $\omega$, $\omega^*$, or $\ZZ$. Such a graph need
not be an exact successor graph: no linear order on two elements has
both vertices isolated.

\begin{lemma}[Dense completion]\label{lem:dense}
A countable functional and inverse-functional directed graph with no
directed cycles and infinitely many components is the exact successor
graph of a linear order.
\end{lemma}
\begin{proof}
Index the components bijectively by $\QQ$. Order elements first by
their component indices and then by their positions along the path. Consecutive elements of one
component are exactly its edges. Two elements from different components
have an intervening component, so they cannot be consecutive.
\end{proof}

\section{Star types of bounded multiplicity}\label{sec:star-types}
We first identify the part of a model that our certificate will store
element by element. The test uses star types, defined below, which
record an element’s 1-type and complete successor neighbourhood. Star
types with the same source 1-type may have different neighbourhoods
and therefore different multiplicity bounds.

\begin{definition}
A \emph{ray} $(\tau,M,\gamma)$ describes one neighbour: its
1-type $\tau$, a nonempty legal mask $M$, and the binary assignment
$\gamma$ from the source to that neighbour. A \emph{star type} $p$
consists of a source 1-type $\tp(p)$ and a set $\Rays(p)$ of rays
whose masks are pairwise disjoint. Distinct rays stand for distinct
neighbours. A star type is \emph{legal} if
$\alpha[\tp(p);=]$ holds and
$\alpha[\tp(p),\tau,M,\gamma]$ holds for every
$(\tau,M,\gamma)\in\Rays(p)$.
\end{definition}

An element realizes $p$ if $p$ describes its 1-type and all its
neighbours exactly. The \emph{multiplicity} of $p$ in a model is the
number of elements realizing $p$. There are at most $d$ rays, since
each uses at least one of the $d=2k$ ports and their masks are disjoint.
Put $v=u+a$. Encode a star type by its $v$ source bits and $d$ ray
slots, ordering the rays by their first port. Each slot stores an
occupancy bit, a mask ($d$ bits), a neighbour 1-type ($v$ bits), and
a binary assignment ($2a$ bits). Hence the number $H$ of legal star
types satisfies
\begin{equation}
 H\leq 2^{\,v+d(1+d+v+2a)}.
 \label{eq:star-types}
\end{equation}
Counting all bit strings includes invalid descriptions and only enlarges
the bound. Its exponent is polynomial in $n$, so $H\leq2^{n^{O(1)}}$.

Guess a nonempty set $P$ of legal star types, intended to be precisely the
realized star types, and write $T(X)=\{\tp(p):p\in X\}$ for $X\subseteq P$.
The inverse of a ray $e=(\tau,M,\gamma)$ at a source of star type
$p$ is $(\tp(p),M^{-1},\Inv(\gamma))$. The star types that
\emph{match} $e$ form the set
\begin{equation}
 \Match_P(p,e)=\{q\in P:\tp(q)=\tau,\
   (\tp(p),M^{-1},\Inv(\gamma))\in\Rays(q)\}.
 \label{eq:match}
\end{equation}
Reject immediately if a matching set is empty.

There are two ways to force a star type to have bounded multiplicity.
First, if its source 1-type admits no far assignment with a realized
1-type $\tau$, each realization is either equal to or a neighbour of
one fixed element of 1-type $\tau$.
Second, if a required neighbour can only have a bounded star type,
injectivity of the relevant successor port bounds the number of sources
as well. These two observations give the closure
\begin{align}
 B_0&=\{p\in P:\exists\tau\in T(P),\ A_{\tp(p)\tau}=\varnothing\},
 \notag\\
 B_{r+1}&=B_r\cup\{p\in P:\exists e\in\Rays(p),\
                         \Match_P(p,e)\subseteq B_r\}.
 \label{eq:closure}
\end{align}
The sequence stabilizes after at most $|P|$ strict increases, each of
which adds a previously absent star type.
Let $B$ be its fixed point and $U=P\setminus B$. Call these star types
\emph{bounded} and \emph{regular}, respectively. ``Regular'' means that
the star type survives this closure; it need not occur infinitely often
in a given model. An element is called bounded or regular according to
its star type.

\begin{lemma}[Multiplicity bounds]\label{lem:bounds}
For each $p\in B$ one can compute a bound $K_p$ on its multiplicity in
every model realizing exactly $P$, with
\begin{equation}
 \sum_{p\in B}K_p\leq(d+1)2^{|P|}.
 \label{eq:sum}
\end{equation}
For each $p\in U$, all sets $A_{\tp(p)\tau}$ with $\tau\in T(P)$
are nonempty, and every ray of $p$ has a match in $U$.
\end{lemma}
\begin{proof}
For $p\in B_0$, choose a realized 1-type $\tau$ with
$A_{\tp(p)\tau}=\varnothing$, and fix one element of 1-type
$\tau$. Every element of 1-type
$\tp(p)$ equals this element or is its neighbour. Set $K_p=d+1$.

When a new star type $p$ enters the closure, choose a ray $e$ all of
whose matches were added earlier, and set
$K_p=\sum_{q\in\Match_P(p,e)}K_q$.
Map each realization of $p$ to its $e$-neighbour. This map is injective:
choose any port in the nonempty mask of $e$; two sources with one target
would occupy the same inverse port there. The stated bound is therefore
valid.

Process the bounded types in closure order, maintaining $d+1$ plus
the sum of bounds assigned so far. A seed contributes $d+1$; any later
contribution is a sum of earlier bounds. Neither exceeds the current
total, so each step at most doubles it. After at most $|P|$ steps the
total is at most $(d+1)2^{|P|}$ and dominates the sum
in~\eqref{eq:sum}, also when $B=\varnothing$.
The assertions about $U$ follow from the fixed point.
\end{proof}

\begin{example}\label{ex:closure-small}
Suppose $\alpha$ forbids two distinct $K$-elements and $K$ is
realized. Every star type whose source 1-type contains $K$ is bounded: two
distinct elements of that 1-type admit no far assignment.
A star type requiring a $K$-element as a predecessor is bounded next. Another star type
with the same source 1-type, but allowing a regular predecessor,
can remain regular. The closure must therefore distinguish star types
with the same 1-type. Appendix~\ref{app:example} gives a four-star-type
example tracing this distinction through the full construction.
\end{example}

If $U=\varnothing$, Lemma~\ref{lem:bounds} already bounds the whole
domain, and a finite model can be guessed explicitly. For the rest of
the construction assume $U\ne\varnothing$.

\section{The finite certificate}\label{sec:certificate}
The certificate stores a finite interface and describes other elements
by \emph{rows}: star types with assignments to the bounded elements.
Section~\ref{sec:realization} constructs a model from these data and
then proves that every model supplies them.

\subsection{The finite interface}
Guess a finite set $C$ with labels $\lambda:C\to P$. Its
\emph{core} $C_B=\{c\in C:\lambda(c)\in B\}$ will contain all
bounded elements. The \emph{interface} $C$ consists of this
core and its neighbours; $C\setminus C_B$ is the \emph{boundary}.
There are at most $\sum_{p\in B}K_p$ core elements. Each contributes
itself and at most $d$ neighbours, so Lemma~\ref{lem:bounds} and
$|P|\leq H$ give
\begin{equation}
 |C|\leq(d+1)\sum_{p\in B}K_p\leq(d+1)^2 2^H.
 \label{eq:court}
\end{equation}
Guess an \emph{interface pair table} giving a mask and a binary assignment on
each unordered pair of distinct elements of $C$, in one orientation.
Reverse data are obtained by inversion; diagonals are fixed by the labels.
Require:
\begin{enumerate}[({I}1)]
\item Each $p\in B$ labels between one and $K_p$ elements, and every
boundary element has a neighbour in the core.
\item At each $c\in C$, the neighbours specified by the pair table realize distinct rays
of $\lambda(c)$. At each $c\in C_B$, they realize all its rays.
\item Every pair satisfies $\alpha$, and each colour graph on $C$ is
acyclic.
\end{enumerate}
In (I2), an entry $(M,\gamma)$ from $c$ to $c'$ with
$M\ne\varnothing$ realizes the ray
$(\tp(\lambda(c')),M,\gamma)$ at $c$ and its inverse at $c'$.
Thus the check fixes both endpoints and forbids reuse of a ray.

For an element $a$ assigned star type $p$, a \emph{ray occurrence}
is a pair $(a,e)$ with $e\in\Rays(p)$. It is \emph{free} if no
neighbour has yet been assigned to it. Every free occurrence in $C$ lies on the
boundary and has a match in $U$, by Lemma~\ref{lem:bounds}.
No new element will have a neighbour in the core. If $B=\varnothing$,
then $C=\varnothing$.
The table records the restriction of the successor relations to $C$.
It does not take successors afresh in a restricted order: deleting
intervening elements must not create new edges.

\subsection{Rows and witnesses}\label{sec:rows}
A \emph{row} describes a new regular element \(z\), far from the whole
core. Formally, it is a pair \((q,f)\), where \(q\in U\) is the star
type of \(z\) and \(f:C_B\to\Gamma\) satisfies \(f(c)\in
A_{\mathrm{tp}(q),\mathrm{tp}(\lambda(c))}\) for every \(c\in C_B\).
Thus one row fixes the binary assignment to every core element.  The
orientation is from \(z\) to \(c\): \(R\in f(c)\) means \(R(z,c)\),
and \(R^{-1}\in f(c)\) means \(R(c,z)\).

A row is \emph{feasible} if, for each witness obligation $j$, at least one
of the following choices makes $\beta_j$ true:
\begin{enumerate}[(W1)]
\item the source itself: $\beta_j[\tp(q);=]$;
\item a neighbour described by $(\tau,M,\gamma)\in\Rays(q)$:
$\beta_j[\tp(q),\tau,M,\gamma]$;
\item a far regular element: some $\tau\in T(U)$ and
$\gamma\in A_{\tp(q)\tau}$ satisfy
$\beta_j[\tp(q),\tau,\varnothing,\gamma]$;
\item a core element $c$:
$\beta_j[\tp(q),\tp(\lambda(c)),\varnothing,f(c)]$.
\end{enumerate}
The same map $f$ must work for all obligations, although their witnesses
may differ. Guess a finite set $\mathcal F$ of feasible rows, containing
at least one row for every $q\in U$, with
\begin{equation}
 |\mathcal F|\leq |U|+m|C|.
 \label{eq:rows}
\end{equation}
The bound concerns the selected rows, not all feasible rows.  It
allows one row for each regular star type and at most one additional
row per each of the $m$ interface obligations;
Lemma~\ref{lem:extraction} proves that these suffice.  Different rows
with the same star type may be necessary: if \(E\) is asymmetric and
\(c\) lies in the core, one copy \(z\) may require \(E(z,c)\), while
another must witness \(E(c,z)\).  These requirements cannot share a
row.  Every selected row receives infinitely many copies.
Appendix~\ref{app:example} gives an explicit example.

For each $c\in C$ and obligation $j$, choose a witness in one of
three ways:
\begin{enumerate}[(a)]
\item an element of $C$, possibly $c$ itself;
\item the neighbour to be assigned to a free ray occurrence at $c$;
\item a new far element described by a row $(q,f)\in\mathcal F$,
with an assignment $\gamma\in A_{\tp(\lambda(c)),\tp(q)}$ from $c$
to that element.
\end{enumerate}
The chosen diagonal, pair, or ray must satisfy $\beta_j$. In (c), this
means $\beta_j[\tp(\lambda(c)),\tp(q),\varnothing,\gamma]$; if
$c\in C_B$, also require
\begin{equation}
 \gamma=\Inv(f(c)).
 \label{eq:pin}
\end{equation}
The equality requires the two descriptions of the pair to agree.
For a boundary source there is no row entry to constrain $\gamma$.

\begin{samepage}
After $P$ has been chosen, the certificate therefore has three parts:
\begin{center}
\renewcommand{\arraystretch}{1.15}
\begin{tabular}{@{}p{.27\textwidth}p{.65\textwidth}@{}}
\toprule
Data & What they specify \\
\midrule
Interface and pair table & The bounded elements, their
neighbours, and all pairs between them. \\
Rows $\mathcal F$ & The allowed descriptions of new elements:
a star type and one consistent assignment to the whole core. \\
Witness choices & For every interface obligation, an element of $C$,
a free ray occurrence, or a row and a compatible far assignment. \\
\bottomrule
\end{tabular}
\end{center}
\end{samepage}
A certificate is \emph{accepting} if all these conditions hold.

\section{Certificates and models}\label{sec:realization}
Sections~\ref{sec:heights}–\ref{sec:far} construct an admissible
model from an accepting certificate; Section~\ref{sec:extracting}
proves the converse by extracting a certificate from an admissible
model. For the construction, fix an accepting certificate with
\(U\ne\varnothing\). Give every row infinitely many copies. We first
match their successor requirements, then fill the far pairs to satisfy
the remaining obligations.

\subsection{Matching rays at integer heights}\label{sec:heights}
Choose height vectors $\mathbf h(c)=(h_1(c),\ldots,h_k(c))$ on $C$
so that $h_i$ increases by one along each $S_i$-edge. This is possible
because each colour graph is a disjoint union of paths. Add
\begin{equation}
 R=\mathcal F\times\ZZ^k\times\NN\times\NN,
 \label{eq:regular}
\end{equation}
where $\NN=\{0,1,2,\ldots\}$. An element \((r,\mathbf{z},g,\ell)\)
carries row \(r=(q,f)\) and has star type \(q\). The vector
\(\mathbf{z}\) gives its heights and will rule out cycles. Elements
with the same \(g\) form one group: successor edges inside \(R\) stay
in that group, while far witnesses are chosen in group
\(g+1\). Finally, \(\ell\) supplies infinitely many copies of each row
at every height and group.

For a legal mask $M$, define
\[
 \delta(M)_i=
 \begin{cases}
 1,&i^+\in M,\\
 -1,&i^-\in M,\\
 0,&\text{otherwise}.
 \end{cases}
\]
A ray with mask $M$ will lead from height $\mathbf z$ to
$\mathbf z+\delta(M)$. For example, $\{1^+,2^-\}$ leads to one
shared neighbour at displacement $(1,-1)$, whereas separate rays
$\{1^+\}$ and $\{2^-\}$ lead to different heights.
The map $M\mapsto\delta(M)$ is injective,
$\delta(M^{-1})=-\delta(M)$, and nonempty masks have nonzero displacement.

\paragraph*{Attach the interface.}
For each free occurrence $(c,e)$ with mask $M$, choose a row whose
star type matches $e$, and a distinct copy at height
$\mathbf h(c)+\delta(M)$ in group $0$. Match $e$ to its inverse
there, using its binary assignment. This uses just one occurrence at
the copy. Every such attachment starts at
the boundary, since all core rays are already realized.

For each interface witness choice of kind (c), reserve a separate copy
of its chosen row in group \(0\), distinct from all attachment
targets. It receives no successor edge to \(C\), so the chosen pair is
far. Separate copies also prevent different obligations from
prescribing incompatible assignments to the same pair. There are at
most \(d|C|\) attachment targets and \(m|C|\) reserved witnesses, so
the coordinate \(\ell\) supplies enough copies. All of \(R\) remains
far from the core.

\begin{lemma}[Ray realization]\label{lem:matching}
All remaining free ray occurrences in $R$ can be matched so that each assigned
star type is realized exactly, every colour is acyclic, and edges
inside $R$ preserve $g$.
\end{lemma}
\begin{proof}
An occurrence of $(\tau,M,\gamma)$ at a source of 1-type $\sigma$
prescribes 2-type $\theta=(\sigma,\tau,M,\gamma)$. Its inverse is
$\theta^{-1}=(\tau,\sigma,M^{-1},\Inv(\gamma))$.
They are distinct because $M$ is nonempty and legal. Choose one
representative from each inverse pair of 2-types occurring in regular
star types.

For each representative $\theta$ with mask $M$ and each
$(\mathbf z,g)$, match its free occurrences at $(\mathbf z,g)$
bijectively with inverse occurrences at $(\mathbf z+\delta(M),g)$,
and give each matched pair its specified 2-type.
Both sets are countably infinite. Regularity ensures that both
orientations occur in $U$, and $\mathcal F$ contains a row for every
star type in $U$. Each such row has infinitely many copies at every
fixed height and group, and attachments consume only finitely many
occurrences.
Every occurrence belongs to exactly one such matching: one of inverse
type at height $\mathbf w$ belongs to the matching at
$\mathbf w-\delta(M)$.

No pair is a loop, since its displacement is nonzero. If two
occurrences at one element reached the same neighbour, their masks
would have the same displacement, hence would be equal. Distinct rays
cannot have equal nonempty masks, and each occurrence is used only
once. Thus no pair is assigned twice, from either endpoint. Distinct
attachment targets give the same property for pairs involving
$C$. Since each port belongs to at most one ray and each occurrence is
used only once, every colour graph is functional and
inverse-functional.

Put $H_i(c)=h_i(c)$ and $H_i(r,\mathbf z,g,\ell)=z_i$.
Every $S_i$-edge, including an attachment, increases $H_i$ by one,
so no $S_i$-cycle exists. All matchings inside $R$ preserve $g$.
\end{proof}

Every group $g\geq1$ is nonempty and has no successor edge leaving it.
Hence each colour has infinitely many components.
Lemma~\ref{lem:dense} supplies an order with exactly that successor
relation, independently for each colour.

\subsection{Assigning far pairs}\label{sec:far}
For every new element with row $(q,f)$, assign $f(c)$ on its pair
to each $c\in C_B$. Also assign each reserved interface-witness pair
as specified by its witness choice. The two instructions agree at a
core source by~\eqref{eq:pin}; at a boundary source they concern
separate pairs. Distinct reserved witnesses prevent further conflicts.
Fill the other far pairs between $C$ and $R$ from their permitted
assignment sets. These are nonempty because the endpoint in $R$ is
regular, even when the endpoint in $C$ is bounded.

All interface obligations are now satisfied. For new elements, any
obligation admitting (W1), (W2), or (W4) already has its diagonal,
ray, or core witness. By row feasibility, every remaining obligation
admits (W3), the case of a far regular witness (Section~\ref{sec:rows}).

\begin{lemma}[Far witnesses]\label{lem:witnesses}
The remaining obligations can be satisfied without changing any assigned
pair or violating $\alpha$.
\end{lemma}
\begin{proof}
For each remaining obligation of an element in group $g$, choose a
target 1-type and a permitted assignment as in (W3), and a target of that
type in group $g+1$. Use distinct targets for different obligations
of the same source. Every required 1-type has infinitely many realizations
there, and every selected pair is far because successor edges preserve
groups.

No unordered pair is selected twice. Two different sources could do so
only by selecting each other, whereas selections always increase the
group number. Sharing a target uses different pairs and causes no
conflict. Assign the selected pairs, then fill every remaining far pair
in $R$ from its set of permitted assignments. These sets are nonempty by
regularity. Choose one orientation per pair; symmetry of $\alpha$
validates the reverse assignment.
\end{proof}

\subsection{From a model to a certificate}\label{sec:extracting}
\begin{lemma}[Extraction]\label{lem:extraction}
Every admissible model realizing exactly $P$, with $U\ne\varnothing$,
supplies an accepting certificate.
\end{lemma}
\begin{proof}
Take all bounded elements as the core, add their neighbours
to obtain $C$, and retain the actual pair table. Lemma~\ref{lem:bounds}
and the successor semantics give (I1)--(I3) and~\eqref{eq:court}.

For each \(q\in U\), choose a realization \(z\), possibly on the
boundary, and form a row \((q,f)\) describing a new element far from
the core. Copy the assignments of \(z\) to far core elements. At a
core neighbour, use any permitted far assignment, which exists by
regularity. A neighbour witness is still represented by its ray in
\(q\), and that ray has a matching regular star type by
Lemma~\ref{lem:bounds}. Thus actual witnesses give (W1) on the
diagonal, (W2) at neighbours, (W4) at far core elements, and (W3)
otherwise, proving feasibility.

Choose actual witnesses for the interface obligations. A witness in
$C$ gives (a); a neighbour outside $C$ gives (b). Every remaining
witness lies outside $C$ and is far from the whole core. Add its actual
row to $\mathcal F$ and use (c). The same classification proves row
feasibility, and the actual assignments satisfy~\eqref{eq:pin}.
There is one initial row per regular type and at most one additional
row per interface obligation, proving~\eqref{eq:rows}.
\end{proof}

\begin{proposition}[Certificate equivalence]\label{lem:certificate}
Fix $P$ with nonempty matching sets and $U\ne\varnothing$.
An accepting certificate exists if and only if $\Phi$ has an admissible
model realizing exactly $P$. Every accepting certificate yields a
countably infinite such model, with every star type in $U$ realized
infinitely often.
\end{proposition}
\begin{proof}
Lemmas~\ref{lem:matching} and~\ref{lem:dense} realize the certified
star types with exact successors. Star type legality, the interface
checks, and the permitted far assignments ensure $\alpha$ everywhere.
The witness choices and Lemma~\ref{lem:witnesses} satisfy all
obligations. The interface realizes every star type in $B$, and the rows
supply infinitely many realizations of every star type in $U$.
The converse is Lemma~\ref{lem:extraction}.
\end{proof}

Infinitely many copies allow the ray matchings, and infinitely many
components allow dense completion. The construction does not in general
yield a finite model or orders of type $\omega$.

\section{Decision procedures and consequences}\label{sec:decision}
\subsection{Satisfiability and the complexity bound}
The algorithm enumerates the legal star types and the permitted
assignment sets $A_{\sigma\tau}$, guesses
a nonempty set $P$ of realized star types, rejects empty matching sets,
and computes the boundedness closure $B$ from~\eqref{eq:closure}.
Set $U=P\setminus B$, the set of regular star types. There are two branches.
\begin{enumerate}[(1)]
\item If $U=\varnothing$, guess a nonempty structure of size at most
$\sum_{p\in P}K_p$, with the multiplicity bounds $K_p$ from
Lemma~\ref{lem:bounds}. Check $\Phi$, exact realization of $P$, and that
each distinguished successor is a single directed spanning path.
\item If $U\ne\varnothing$, guess and check an accepting certificate.
\end{enumerate}
The first branch covers all models with its star type set by
Lemma~\ref{lem:bounds}. The second is sound and complete by
Proposition~\ref{lem:certificate}. The second branch can
also accept a star type set originally extracted from a finite model;
it then constructs an infinite model with the same set of star types.

\begin{proof}[Proof of Theorem~\ref{thm:main}]
Recall that $H$ counts legal star types, $d=2k$, $a$ counts additional
binary predicates, and $m$ counts the witness obligations in~\eqref{eq:scott}.
The lists of legal star types and permitted assignments have singly
exponential size by~\eqref{eq:star-types} and $|\Gamma|=2^{2a}$.
Once $P$ is chosen, matching sets and at most $H$ closure steps are
computed in time polynomial in these lists. By~\eqref{eq:sum},
$K_p\leq(d+1)2^H$, so its binary representation has
$O(H+\log(d+1))$ bits. Thus even arithmetic on the multiplicity
bounds takes only singly exponential time.

Put $c=|C|$ and $c_B=|C_B|$. Assignments in the pair table and rows take
\[
 2a\left(\binom c2+|\mathcal F|c_B\right)\quad\text{bits}:
\]
one $2a$-bit assignment per interface pair and per row entry at a core
element. Here $c_B\leq c$ and $|\mathcal F|\leq H+mc$ by~\eqref{eq:rows}.
Masks use $d$ bits; labels and indices use binary encodings. Together
with the $mc$ witness choices, their total size is polynomial in $n,H,c$.
All checks are polynomial in the explicit certificate and the precomputed
lists; for (W3), scan the realized regular 1-types and permitted assignments.
Only selected rows are checked. By~\eqref{eq:court}, $c$ is exponential
in $H$, which is singly exponential in $n$. Hence the certificate,
and the time to guess and check it, are at most doubly exponential:
$2^{2^{n^{O(1)}}}$. The finite-model branch obeys the smaller domain
bound~\eqref{eq:sum} and the same time bound.

Hardness follows by adding an unused successor to ordinary $\FO$.
\end{proof}

\subsection{Infinite models and an effective size threshold}
\label{sec:consequences}
The branch with $U\ne\varnothing$ also decides infinite
satisfiability. The multiplicity bounds give a threshold above which a
finite model implies the existence of an infinite model.

\begin{theorem}[Infinite satisfiability]\label{thm:infinite}
Existence of an infinite admissible model belongs to $\TwoNEXP$ and is
$\NEXP$-hard already for one successor. Every sentence with an infinite
admissible model has a countably infinite admissible model supplied by
an accepting certificate with $U\ne\varnothing$.
\end{theorem}
\begin{proof}
In an infinite model, the finitely many bounded star types account for
only finitely many elements. Its realized star type set therefore has
$U\ne\varnothing$, and Lemma~\ref{lem:extraction} supplies a
certificate. Conversely, every such certificate yields an infinite
model. Restricting the algorithm to its second branch proves the bound.

For hardness, relativize ordinary $\FO$ to a nonempty unary predicate
and add infinitely many elements outside it.
\end{proof}

\begin{corollary}[Effective threshold]\label{cor:threshold}
Let $H$ be the number of legal star types of a normalized sentence $\Phi$,
let $d=2k$, and put $F_\Phi=(d+1)2^H$. If $\Phi$ has a finite admissible model of
size greater than $F_\Phi$, it has an infinite admissible model.
Equivalently, if it has no infinite admissible model, every admissible
model has size at most $F_\Phi\leq2^{2^{n^{O(1)}}}$.
\end{corollary}
\begin{proof}
Take the realized star type set $P$ of the given finite model. If all its
star types were bounded, Lemma~\ref{lem:bounds} would bound its size by
$(d+1)2^{|P|}\leq(d+1)2^H=F_\Phi$, using $|P|\leq H$.
Hence $U\ne\varnothing$;
extracting and realizing a certificate gives an infinite model.
The second exponential comes from exponentiating the singly exponential
number $H$ of legal star types.
\end{proof}

This is not a small-model property for finite satisfiability: it bounds
finite models only when there is no infinite model. In that case,
checking all structures up to $F_\Phi$ determines all possible model
sizes. Normalization preserves domain size, so the threshold also
applies to the original sentence.

\paragraph*{How large can the core be?}
An explicit core cannot always have singly exponential size.
Proposition~\ref{prop:large-core} in Appendix~\ref{sec:large-core}
gives sentences polynomial in $r$ forcing exactly $2^{2^r}-1$ bounded
elements, using only two successors and unary predicates; all their
models are infinite. This is a model-size lower bound, not
$\TwoNEXP$ hardness. 

\section{Finite models with bounded adjacency disagreement}
\label{sec:budget}
The preceding construction relies on infinitely many copies and a dense
ordering of components. For finite models, each successor must instead
form one spanning path. We now obtain a separate decision procedure
by bounding the number of reference adjacencies lost by the other paths.
Throughout this section all models are finite.

Define the \emph{disagreement cost}
relative to the first order by
\begin{equation}
 \Dis(\mathcal A)=\sum_{i=2}^{k}
             |S_1^{\mathcal A}\setminus S_i^{\mathcal A}|.
 \label{eq:cost}
\end{equation}
Budgeted finite satisfiability takes a sentence $\varphi$ and a
nonnegative integer $b$, and asks whether there is a finite admissible
model with $\Dis(\mathcal A)\leq b$. The orders, including the reference
order, are chosen with the model. Since every finite successor has
$|A|-1$ edges, each summand is also the number of added edges and half
the symmetric difference. A reference edge lost by two orders is counted
twice.

\begin{theorem}[Budgeted finite satisfiability]\label{thm:budget}
Budgeted finite satisfiability is $\NEXP$-complete for a binary-encoded
budget, uniformly in the number of successors and with arbitrary
additional unary and binary predicates. It is also $\NEXP$-complete
for unary-encoded budgets and for every fixed $b\geq0$.
Every satisfiable instance has a model of size at most
$2^{p(|\varphi|+k+b)}$, for a fixed polynomial $p$.
\end{theorem}
The model-size bound depends on the value of $b$, which can be
exponential in its binary encoding length. The decision-time bound uses
that length: the algorithm below need not list a model.

\subsection{Reduction to one successor}
We use two-variable logic with counting quantifiers, $\mathrm C^2$.
The quantifier $\exists^{\leq t}y\,\theta$ means that at most $t$
distinct values of $y$ satisfy $\theta$; $\exists^{=t}$ is defined
similarly. Counts are written in binary.
Let $\mathrm C^2[S]$ denote this logic with one
distinguished successor, interpreted as a single finite spanning path,
and arbitrary additional unary and binary predicates. Its finite
satisfiability problem is $\NEXP$-complete for binary-encoded counts.
The upper bound follows from~\cite[Theorem~1]{charatonikEtAl2018}:
use the distinguished forest relation for $S$, bound its outdegree by
one, and require exactly one root (an element with no incoming edge).
Every finite forest satisfying these counting conditions is a spanning
path, including the one-element path. The lower bound follows from ordinary
$\FO$-satisfiability.

Keep $S_1$ distinguished and treat the other $S_i$ as ordinary predicates.
Their ranks stay constant on shared $S_1$-edges and increase on added
edges, excluding cycles since $S_1$ is acyclic. Degree and root
constraints then make each $S_i$ a spanning path.

Guess nonnegative integers $b_2,\ldots,b_k$ with
$\sum_{i=2}^k b_i\leq b$, all written in binary.
An admissible model of cost at most $b$ has at most $b$ added edges per path,
so ranks $0,\ldots,b$ suffice.
Put $w=\lceil\log_2(b+2)\rceil$: then $2^w>b$, with at least one
bit even when $b=0$. For each $i\geq2$, introduce fresh unary
predicates $U_{i,0},\ldots,U_{i,w-1}$ encoding a rank
$\rho_i(a)\in\{0,\ldots,2^w-1\}$, with bit $0$ least significant.
The quantifier-free comparisons are
\[
\begin{aligned}
 \operatorname{Eq}_i(x,y)
   &:=\bigwedge_{h<w}(U_{i,h}(x)\leftrightarrow U_{i,h}(y)),\\
 \operatorname{Lt}_i(x,y)
   &:=\bigvee_{h<w}\left(\neg U_{i,h}(x)\land U_{i,h}(y)
       \land\bigwedge_{h<\ell<w}
          (U_{i,\ell}(x)\leftrightarrow U_{i,\ell}(y))\right).
\end{aligned}
\]
They express equality and strict comparison of the ranks. Equality
compares $w$ bits; strict comparison has $w$ cases, each inspecting at
most $w$ bits. Thus the expansions have $O(w^2)$ atom occurrences.

For each $i\geq2$, include the degree and root conditions
\begin{equation}
 \forall x\exists^{\leq1}y\,S_i(x,y),\qquad
 \forall x\exists^{\leq1}y\,S_i(y,x),\qquad
 \exists^{=1}x\,\neg\exists y\,S_i(y,x),
 \label{eq:budget-degree}
\end{equation}
and the rank condition
\begin{equation}
 \forall x\forall y\left(S_i(x,y)\to
 \left((S_1(x,y)\land\operatorname{Eq}_i(x,y))
 \lor(\neg S_1(x,y)\land\operatorname{Lt}_i(x,y))\right)\right).
 \label{eq:budget-rank}
\end{equation}
Finally, use a fresh unary predicate $Q_i$ to count the tails of lost
reference edges:
\begin{equation}
 \forall x\left(Q_i(x)\leftrightarrow
       \exists y\,(S_1(x,y)\land\neg S_i(x,y))\right),
 \qquad \exists^{\leq b_i}x\,Q_i(x).
 \label{eq:budget-count}
\end{equation}
Let $\psi_{\mathbf b}$ be the conjunction of $\varphi$ with
\eqref{eq:budget-degree}--\eqref{eq:budget-count} for all $i\geq2$,
expanding the comparisons. Each of the $k-1$ relations adds $O(w^2)$
atom occurrences and a count $b_i\leq b$, written with at most $w$ bits.
Thus the resulting $\mathrm C^2[S_1]$ sentence has size polynomial in
the input length $N=|\varphi|+k+w$.

\subsection{Correctness and complexity}
Suppose first that $\psi_{\mathbf b}$ has a finite model. On an
$S_i$-cycle the rank cannot decrease and increases strictly on any edge
outside $S_1$. Hence such a cycle would consist entirely of $S_1$-edges,
contrary to the semantics of the distinguished successor.
Together with~\eqref{eq:budget-degree}, acyclicity makes $S_i$ a finite
union of directed paths, including possible isolated vertices. Each
component has one vertex with no predecessor, and there is exactly one
such vertex in total. Thus $S_i$ is a single spanning path.
Since $S_1$ is functional, $|Q_i|=|S_1\setminus S_i|$.
Equation~\eqref{eq:budget-count} therefore gives total cost at most
$\sum_{i=2}^k b_i\leq b$, and the original-signature reduct is an
admissible model of $\varphi$.

Conversely, let $\mathcal A$ be a finite admissible model of cost at
most $b$. Choose $b_i=|S_1\setminus S_i|$. Both paths have $|A|-1$
edges, so $|S_i\setminus S_1|=b_i$. For an element $a$, let
$\rho_i(a)$ count the edges in $S_i\setminus S_1$ on the $S_i$-path
from its first element to $a$. These ranks lie between $0$ and
$b_i\leq b<2^w$ and therefore fit into the chosen bits. They remain
constant on shared edges and increase by one on the other edges,
satisfying~\eqref{eq:budget-rank}. The degree, root, and cost conditions
hold as well. This expands $\mathcal A$ to a finite model of
$\psi_{\mathbf b}$.

Each guessed $b_i\leq b$ takes at most $w$ bits, so all guesses use
$O(kw)=O(k\log(b+2))$ bits. Applying the
$\NEXP$ procedure for finite satisfiability of $\mathrm C^2[S_1]$
to the polynomial-size sentence $\psi_{\mathbf b}$ gives time
$2^{N^{O(1)}}$. This also covers $k=1$, when no extra constraints
are needed, and $b=0$, when all ranks in the completeness argument
are zero. The procedure need not list a model explicitly.

For hardness, identical successors have cost zero, so the usual
$\FO$ reduction works for every fixed $b\geq0$ and $k\geq1$.

\subsection{The finite-model size bound}
To prove the size bound, we encode the successors by at most
\(b+1\) common blocks and apply a small-model theorem for one order.
Fix a finite admissible model \(A\models\varphi\) of cost at most \(b\).
Let \(L=\bigcup_{i=2}^k(S_1\setminus S_i)\).
Since the union counts each lost edge once, \(|L|\le D(A)\le b\).
Cutting \(L\) in the reference path gives
\(r=|L|+1\le b+1\) nonempty directed blocks.
Their internal edges belong to every \(S_i\), so functionality
and inverse-functionality force each block to remain consecutive
in every order. Thus every order permutes the blocks without
reversing them.
Number the blocks in reference order, and let
\(\pi_1=\mathrm{id},\pi_2,\ldots,\pi_k\) be their permutations.
Only block boundaries contribute to the cost.
Writing
\(\mathrm{Adj}(\pi)=\{(\pi(t),\pi(t+1)):1\le t<r\}\),
we have
\begin{equation}
 |S_1\setminus S_i|
   =(r-1)-|\Adj(\pi_i)\cap\Adj(\mathrm{id})|.
 \label{eq:block-cost}
\end{equation}
For example, three blocks can be arranged as follows:
\begin{center}
\begin{tabular}{@{}lcc@{}}
\toprule
Order & Block sequence & Lost reference adjacencies \\
\midrule
Reference & $B_1\,B_2\,B_3$ & --- \\
Second & $B_2\,B_3\,B_1$ & $B_1\to B_2$ \\
Third & $B_3\,B_1\,B_2$ & $B_2\to B_3$ \\
\bottomrule
\end{tabular}
\end{center}
The cost is two, independently of block lengths.
Keeping the permutations fixed will preserve this cost when
we apply the small-model theorem.

We now encode these fixed permutations using one linear order
\(<\) and its successor \(S\).
Fresh unary predicates \(B_1,\ldots,B_r\) mark the blocks
in reference order.
Require them to partition the domain into nonempty sets and
put every \(B_a\)-element before every \(B_c\)-element whenever
\(a<c\); this makes each block an interval.
To express the connections between blocks, define predicates
\(F_j,L_j\) for their first and last elements:
\[
\begin{aligned}
 \forall x\bigl(F_j(x)&\leftrightarrow
   (B_j(x)\land\neg\exists y\,(B_j(y)\land y<x))\bigr),\\
 \forall x\bigl(L_j(x)&\leftrightarrow
   (B_j(x)\land\neg\exists y\,(B_j(y)\land x<y))\bigr).
\end{aligned}
\]

The partition and block order use \(O(r^2)\) clauses, and
the endpoint definitions use \(O(r)\).
In a finite model each nonempty block has unique endpoints,
possibly the same element.
In \(\varphi\), replace \(S_1\) by \(S\) and every \(S_i\),
\(i\ge2\), by
\begin{equation}
\begin{split}
 \widehat S_i(x,y)={}&
 S(x,y)\land\bigvee_{j=1}^r(B_j(x)\land B_j(y))\lor\bigvee_{t=1}^{r-1}
        (L_{\pi_i(t)}(x)\land F_{\pi_i(t+1)}(y)).
\end{split}
 \label{eq:block-decoder}
\end{equation}

The first disjunct preserves edges within blocks; the second
joins successive blocks in order \(\pi_i\).
The substitution is quantifier-free, including for reversed
or repeated arguments.
Let \(\psi\) be the substituted sentence with the block and
endpoint clauses.
Marking the blocks and endpoints of \(A\) gives a model of \(\psi\).

Each replacement has \(O(r)\) atom occurrences.
With the \(O(r^2)\) block clauses and \(r\le b+1\),
the length of \(\psi\) is polynomial in \(|\varphi|+k+b\).
The small-model theorem for one order and its successor,
with additional unary and binary predicates~\cite[Section~3]{zeume2016},
gives a finite model of \(\psi\) of size exponential in this length.
Reconstruct the \(S_i\) by concatenating blocks according to
\(\pi_i\), leaving the other predicates unchanged.
The result satisfies \(\varphi\), retains this small model's
domain, and has the original cost by (15).
Hence its size is at most \(2^{p(|\varphi|+k+b)}\),
proving Theorem~13.

Trying all budgets does not decide unrestricted finite
satisfiability: the argument provides no computable stopping
budget for instances with no finite model.

\section{Conclusion}
Arbitrarily many independent successor relations preserve decidability
of $\FO$-satisfiability over arbitrary order types, with a uniform
$\TwoNEXP$ upper bound and unrestricted additional binary predicates.
The proof separates local compatibility from global realization:
a finite interface records bounded multiplicities and their witness
requirements, while integer heights and dense completion realize the
remaining successor structure. This supplies finite certificates despite
the failure of the finite-model property. The same certificates decide
infinite satisfiability and bound every model of a sentence that has
only finite models.

For finite structures, bounding adjacency disagreement gives a
different way to control the interaction of successors.
The resulting problem remains $\NEXP$-complete with both the number
of successors and a binary budget in the input. Thus allowing a
specified amount of disagreement does not increase the complexity
class beyond the one-successor case. Common blocks explain the
structural restriction, while binary ranks make its succinct encoding
algorithmically useful.

Three questions remain. Can the gap between the $\NEXP$ lower bound
and the $\TwoNEXP$ upper bound for general satisfiability be closed?
The forced doubly exponential core limits explicit enumeration, but
does not establish a matching complexity lower bound. Can certificates
be adapted to scattered or well-founded underlying orders, for which
dense completion is unavailable? Finally, can the interaction of three
or more successors be controlled in finite models without a supplied
disagreement budget? Resolving these questions would clarify
which aspects of order type and interaction between adjacencies govern
the boundary of decidability.

% Main text ends here; the title page and references are counted separately.
\label{end:main}
\clearpage
\bibliography{references}
\

\paragraph*{Use of generative AI.}
During the preparation of this work, the authors utilized GPT-6 Astra
to assist with problem-solving, text condensation, and manuscript
editing. Specifically, based on the authors' initial notes, GPT-6
Astra proposed a very complex proof argument to address the
decidability problem; the authors subsequently spent substantial
effort deconstructing, simplifying, and fully comprehending this
argument, ultimately rewriting nearly the entire proof and
independently verifying its mathematical correctness. Additionally,
GPT-6 Astra and GPT-6.1 SOL were employed as supporting tools for
proofreading, stylistic refinement, and text reduction. The authors
have thoroughly reviewed and edited all content and accept full
responsibility for the final accuracy and integrity of the
publication.

\clearpage
\appendix
\section{Scott normalization on the same domain}\label{app:normalization}
Given $\varphi$, introduce fresh unary predicates $\bar U$ and construct
a sentence
\begin{equation}
 \Phi=\forall x\forall y\,\alpha(x,y)\land
       \bigwedge_{j=1}^{m}\forall x\exists y\,\beta_j(x,y)
 \label{a:scott}
\end{equation}
such that, for every original structure $\mathcal A$,
$\mathcal A\models\varphi$ iff some expansion
$(\mathcal A,\bar U)$ satisfies $\Phi$~\cite{gradel1997}.
For example, the definition $U(x)\leftrightarrow\exists y\,\theta(x,y)$
is replaced by
\begin{align*}
 &\forall x\exists y(\neg U(x)\lor\theta(x,y)),\\
 &\forall x\forall y(U(x)\lor\neg\theta(x,y)).
\end{align*}
The clauses express the two implications. Nonemptiness supplies a witness
for the first clause at elements where $U$ is false. For a closed
subformula $\exists y\,\theta(y)$, use a fresh unary predicate $U$ and impose
\[
 \forall x\forall y\,(U(x)\leftrightarrow U(y)).
\]
Its definition is expressed by the same two clauses above, with
$\theta(y)$ in place of $\theta(x,y)$. Nonemptiness ensures that $U$
has one global truth value and that these clauses equate that value
with the existential sentence. Universal quantifiers are handled by
$\forall y\,\theta\equiv\neg\exists y\,\neg\theta$. Subformulas whose
only free variable is $y$ are treated by exchanging $x$ and $y$
simultaneously; no third variable is needed.

Replace quantified subformulas from the inside out, retaining their
separate definitions rather than expanding them at every occurrence.
There is at most one fresh predicate and a constant number of defining
clauses per quantified subformula; no exponential unfolding is needed.
The remaining Boolean combination of predicates naming closed subformulas
is imposed at every element. A clause with one universally quantified
variable can be given a second, unused universal quantifier. An existential
sentence $\exists y\,\zeta(y)$ can be written $\forall x\exists y\,\zeta(y)$,
again using nonemptiness.
Combining the universal clauses gives~\eqref{a:scott} in polynomial size.
The new predicates are chosen in the expansion; the domain and all original
predicate interpretations are unchanged.
Conversely, the defining clauses force the named subformulas to have their
intended truth values, so forgetting the fresh predicates recovers
$\varphi$. Finally, conjoining $\alpha(y,x)$ to $\alpha(x,y)$ changes no
model because both variables are universally quantified.

\section{The bounded part can be doubly exponential}
\label{sec:large-core}
Polynomial-size sentences can force a doubly exponential finite set
in every model, even when all their models are infinite.

\begin{proposition}\label{prop:large-core}
For every $r\geq1$ there is a satisfiable sentence $\theta_r$, of size
polynomial in $r$, with two distinguished successors and only additional
unary predicates, such that every admissible model has exactly
$2^{2^r}-1$ elements in a designated unary predicate $T$.
All these elements realize bounded star types of $\theta_r$.
Every admissible model of $\theta_r$ is infinite.
\end{proposition}
\begin{proof}
Use unary predicates $T,L,D_0,\ldots,D_{r-1}$, with the $D_j$ encoding
a level in binary, least significant bit first. Let $Z(x)$ and $M(x)$
say that all bits are zero and all bits are one, respectively. Let
$I(x,y)$ say that the level of $y$ is the level of $x$ plus one,
without overflow. Explicitly,
\[
 I(x,y)=\neg M(x)\land
 \bigwedge_{j<r}\left(D_j(y)\leftrightarrow
    \left(D_j(x)\mathbin{\oplus}\bigwedge_{\ell<j}D_\ell(x)\right)\right),
\]
where $\oplus$ abbreviates exclusive or and the empty conjunction is
true. There are $r$ clauses, each inspecting at most $r$ bits, so
expansion gives $O(r^2)$ atom occurrences. The other requirements below
also have polynomial size. Require a unique
$T$-element at level zero, called the root, and at least one element
outside $T$.
Constrain the successor edges by
\begin{align*}
 \forall x\forall y\bigl(S_1(x,y)&\to
             (T(x)\land T(y)\land I(x,y)\land L(y))\bigr),\\
 \forall x\forall y\bigl(S_2(x,y)&\to
             (T(x)\land T(y)\land I(x,y)\land\neg L(y))\bigr).
\end{align*}
Every $T$-element below the maximum level must have a successor in
each colour; every $T$-element above level zero must have a predecessor
in at least one colour. These requirements are the two-variable clauses
\begin{align*}
 &\bigwedge_{i=1}^2\forall x\bigl((T(x)\land\neg M(x))\to
                              \exists y\,S_i(x,y)\bigr),\\
 &\forall x\bigl((T(x)\land\neg Z(x))\to
                         \exists y\,(S_1(y,x)\lor S_2(y,x))\bigr).
\end{align*}

The two children of one element are distinct because they disagree on
$L$. Different parents cannot share a child in one colour, by
inverse-functionality, or across colours, by the same $L$-constraint. Conversely,
every positive-level $T$-element is a child of an element on the
preceding level. Starting from the unique root, level $j$ therefore has
exactly $2^j$ elements. The $r$ bits allow levels $0,\ldots,2^r-1$, so
\[
 |T|=\sum_{j=0}^{2^r-1}2^j=2^{2^r}-1.
\]

For satisfiability, take the full binary tree of this height, using
$S_1$ for left-child edges and $S_2$ for right-child edges, and add
countably infinitely many elements outside $T$ with no incident edges.
In each colour this graph is functional, inverse-functional, acyclic,
and has infinitely many components. Lemma~\ref{lem:dense} supplies the
required orders. No finite model is possible: an element outside $T$
has neither a predecessor nor a successor, whereas a finite spanning
path with more than one vertex has no isolated vertex.

The sentence can be written directly in Scott normal form without fresh
unary predicates. Fix any realized star type set $P$. Let $\rho$ be the root's
1-type. Uniqueness of a $T\land Z$ element makes the set
$A_{\rho\rho}$ of permitted far assignments from~\eqref{eq:far-table}
empty, so the root's star type lies in the initial closure set $B_0$
from~\eqref{eq:closure}.
Every realized $T$-star type on level $j>0$ has a predecessor ray,
because the predecessor clause must hold at each realization. All
star types matching that ray lie at level $j-1$: this follows from
the increment constraint on its inverse edge. Induction on $j$
therefore puts every such star type into the boundedness closure $B$.
\end{proof}

Consequently, an explicit interface need not have singly exponential
size. This is a lower bound on the size of a forced finite part, not a
$\TwoNEXP$ hardness result: a more succinct algorithm could represent
that part without listing all its elements.

\section{A worked certificate}\label{app:example}\label{ex:specification}
Take unary predicates $K,A,V$, two successor relations constrained by
$S_1=S_2$, and an additional binary predicate $E$. The unique $K$-element
starts a path alternating $V$- and $A$-elements. The relation $E$ is
asymmetric and may hold only on far pairs. The example exhibits both
boundedness propagation and the need for several rows of one star type.

Let $\operatorname{Part}(x)$ say that exactly one of $K(x),A(x),V(x)$
holds, and abbreviate
\[
 N(x,y)=\bigvee_{i=1}^{2}(S_i(x,y)\lor S_i(y,x)).
\]
Define the following quantifier-free formulas:
\begin{align*}
 \vartheta(x,y)&=((K(x)\lor A(x))\land V(y))\lor(V(x)\land A(y)),\\
 \alpha_0(x,y)&=\operatorname{Part}(x)
       \land((K(x)\land K(y))\to x=y)\\
 &\quad{}\land(S_1(x,y)\leftrightarrow S_2(x,y))\\
 &\quad{}\land(S_1(x,y)\to\vartheta(x,y))\\
 &\quad{}\land(E(x,y)\to(\neg E(y,x)\land\neg N(x,y))).
\end{align*}
Use $\alpha=\alpha_0(x,y)\land\alpha_0(y,x)$ and five witness matrices:
\begin{equation}
\begin{aligned}
 \beta_1(x,y)&=K(y),\\
 \beta_2(x,y)&=S_1(x,y),\\
 \beta_3(x,y)&=K(x)\lor S_1(y,x),\\
 \beta_4(x,y)&=\neg K(x)\lor(A(y)\land E(x,y)),\\
 \beta_5(x,y)&=\neg A(x)\lor(E(x,y)\land(K(y)\lor V(y))).
\end{aligned}
\label{eq:worked-beta}
\end{equation}
Together with $\alpha$, these matrices define $\Phi_\star$ in
Scott normal form, as in~\eqref{eq:scott}.
The first obligation makes $K$ nonempty; uniqueness comes from $\alpha$.
The next two require a successor for every element and a predecessor for
every non-$K$ element. The last two require the $K$-element to have an $E$-edge to an $A$-element, and each $A$-element to have an $E$-edge to the $K$-element or a $V$-element. No auxiliary unary predicates are needed for Scott normalization.

\paragraph*{Star types and closure.}\label{ex:star-types}
Let $\sigma_K,\sigma_A,\sigma_V$ be the three 1-types; their $E$
diagonal is false by asymmetry. Write $M^+=\{1^+,2^+\}$,
$M^-=\{1^-,2^-\}$. The empty assignment to $E$ is $\varnothing$.
Guess $P=\{p_K,p_0,p_A,p_V\}$ as follows.
\begin{center}
\small
\setlength{\tabcolsep}{4pt}
\begin{tabular}{@{}ccl@{}}
\toprule
Star type & 1-type & Rays \\
\midrule
$p_K$ & $\sigma_K$ & $(\sigma_V,M^+,\varnothing)$ \\
$p_0$ & $\sigma_V$ & $(\sigma_K,M^-,\varnothing),\ (\sigma_A,M^+,\varnothing)$ \\
$p_A$ & $\sigma_A$ & $(\sigma_V,M^-,\varnothing),\ (\sigma_V,M^+,\varnothing)$ \\
$p_V$ & $\sigma_V$ & $(\sigma_A,M^-,\varnothing),\ (\sigma_A,M^+,\varnothing)$ \\
\bottomrule
\end{tabular}
\end{center}
Every ray has an inverse match. The only empty set of permitted far assignments is
$A_{\sigma_K\sigma_K}$, since two distinct elements of 1-type $\sigma_K$ are forbidden;
all other such sets contain $\varnothing,\{E\},\{E^{-1}\}$.
Applying the boundedness closure~\eqref{eq:closure}, we obtain
$B_0=\{p_K\}$. The predecessor ray of $p_0$ has only $p_K$
as a match, so $p_0$ is added next. No further star type is added:
the predecessor ray of $p_A$ matches both $p_0$ and $p_V$, its
successor ray matches $p_V$, and both rays of $p_V$ match $p_A$.
Thus the bounded and regular star types are, respectively,
\[
 B=\{p_K,p_0\},\qquad U=\{p_A,p_V\}.
\]
Here $d=2k=4$. The seed bound is $K_{p_K}=d+1=5$; the sole predecessor
match of $p_0$ is $p_K$, so propagation gives $K_{p_0}=K_{p_K}=5$.
These are safe upper bounds: uniqueness and port functionality make
both actual multiplicities one.
In particular, $p_0$ and $p_V$ have the same 1-type but different
boundedness status. The closure must operate on star types, not 1-types.

\paragraph*{The interface and rows.}\label{ex:court}
Choose the bounded core $C_B=\{c,v_0\}$ and the interface
$C=\{c,v_0,a_0\}$ with labels
$p_K,p_0,p_A$, respectively. In both colours, the interface edges are
$c\to v_0\to a_0$. Set $E(a_0,c)$ true and all other $E$-atoms on the interface false. The only free
ray occurrence is the successor ray at $a_0$; it matches $p_V$.
The core budget is $K_{p_K}+K_{p_0}=10$. Allowing each core element
itself and four neighbours gives the interface bound
$(d+1)(K_{p_K}+K_{p_0})=50$ in~\eqref{eq:court}; our interface uses
only three elements and passes the structural checks.

The row bound~\eqref{eq:rows} allows $|U|+m|C|=2+5\cdot3=17$ rows.
We use only the three rows $r_A=(p_A,f_A)$, $r_V=(p_V,f_V)$, and
$r_\star=(p_A,f_\star)$ shown below. In a row with source $z$, the
entry $\{E\}$ at $u$ means $E(z,u)$, while $\{E^{-1}\}$ means
$E(u,z)$.
\begin{center}
\small
\setlength{\tabcolsep}{5pt}
\begin{tabular}{@{}ccccc@{}}
\toprule
Row & Star type & At $c$ & At $v_0$ & Witness for $\beta_5$ \\
\midrule
$r_A$ & $p_A$ & $\{E\}$ & $\varnothing$ & $c$ \\
$r_V$ & $p_V$ & $\varnothing$ & $\varnothing$ & Source itself \\
$r_\star$ & $p_A$ & $\{E^{-1}\}$ & $\varnothing$ & A far $V$-element \\
\bottomrule
\end{tabular}
\end{center}
All three rows are feasible. Obligation $\beta_1$ uses $c$;
$\beta_2,\beta_3$ use the successor and predecessor rays; and
$\beta_4$ holds on the diagonal. For $\beta_5$, row $r_A$ uses
$c$, row $r_V$ uses its source, and row $r_\star$ uses a far
$V$-element with assignment $\{E\}$. Thus the two rows of star
type $p_A$ satisfy the same obligation in different ways.

The $K$-element's obligation $\beta_4$ cannot use $a_0$, since
$E(a_0,c)$ excludes $E(c,a_0)$. Choose a new element $a_\star$
with row $r_\star$ and assignment $\gamma=\{E\}$ from $c$.
The compatibility condition~\eqref{eq:pin} between the core and the row holds because
$f_\star(c)=\{E^{-1}\}=\Inv(\gamma)$. Row $r_A$ would not work
for this choice. Let $v_+$ denote the future successor of $a_0$.
The interface witness choices are:
\[
\begin{array}{c|ccccc}
\text{Source} & \beta_1 & \beta_2 & \beta_3 & \beta_4 & \beta_5\\ \hline
 c   & c & v_0 & c & a_\star & c\\
 v_0 & c & a_0 & c & v_0 & v_0\\
 a_0 & c & v_+ & v_0 & a_0 & c
\end{array}
\]
Here $v_+$ selects a free ray occurrence, and $a_\star$ selects the row
$r_\star$ and assignment $\{E\}$. The other entries select elements
already in $C$.

\paragraph*{Completing the model.}\label{ex:realization}
We use the construction with integer heights and groups from
Section~\ref{sec:heights}.
Set both interface heights to $0,1,2$ at $c,v_0,a_0$.
The height displacement for the common mask is $\delta(M^+)=(1,1)$,
so attach the free ray occurrence
at $a_0$ to a copy $v_+$ of row $r_V$ at height $(3,3)$ in
group $0$. Reserve a distinct copy $a_\star$ of row $r_\star$ at
height $(0,0)$ in that group for the $K$-element's $E$-obligation.
Its row already sets $E(c,a_\star)$, and it receives no successor
edge to the interface. Choose a copy $w$ of row $r_V$ in group $1$
and set $E(a_\star,w)$.
This pair is far because successor edges preserve the group.
Figure~\ref{fig:worked-certificate} shows these finite choices.

% Keep the diagram and its explanation together after the expanded setup.
\newpage
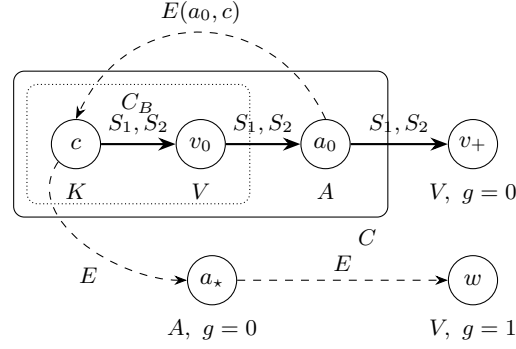
\begin{figure}[H]
\centering
\begin{tikzpicture}[x=1cm,y=1cm,>=Stealth,
  task/.style={circle,draw,minimum size=6.5mm,inner sep=1pt,font=\small},
  every node/.style={font=\footnotesize}]
\node[task] (c) at (0,0) {$c$};
\node[task] (v0) at (1.65,0) {$v_0$};
\node[task] (a0) at (3.3,0) {$a_0$};
\node[task] (vp) at (5.25,0) {$v_+$};
\node[task] (ast) at (1.8,-1.8) {$a_\star$};
\node[task] (w) at (5.25,-1.8) {$w$};
\node[below=2pt] at (c.south) {$K$};
\node[below=2pt] at (v0.south) {$V$};
\node[below=2pt] at (a0.south) {$A$};
\node[below=2pt] at (vp.south) {$V,\ g=0$};
\node[below=2pt] at (ast.south) {$A,\ g=0$};
\node[below=2pt] at (w.south) {$V,\ g=1$};
\draw[->,thick] (c) -- node[above] {$S_1,S_2$} (v0);
\draw[->,thick] (v0) -- node[above] {$S_1,S_2$} (a0);
\draw[->,thick] (a0) -- node[above] {$S_1,S_2$} (vp);
\draw[->,dashed] (a0.north) to[out=115,in=65,looseness=1.3]
  node[above] {$E(a_0,c)$} (c.north);
\draw[->,dashed] (c.south west) to[out=235,in=180]
  node[pos=.63,below left] {$E$} (ast.west);
\draw[->,dashed] (ast) -- node[above] {$E$} (w);
\begin{scope}[on background layer]
\node[draw,densely dotted,rounded corners,fit=(c)(v0),inner xsep=9pt,inner ysep=13pt] (core) {};
\node[draw,rounded corners,fit=(c)(v0)(a0),inner xsep=14pt,inner ysep=18pt] (interface) {};
\end{scope}
\node[below left=1pt] at (interface.south east) {$C$};
\node at (0.82,0.54) {$C_B$};
\end{tikzpicture}
\caption{The interface and its reserved witnesses. The interface is inside the
solid box, and its core is inside the dotted box. Solid arrows are pairs
in both $S_1$ and $S_2$; dashed arrows are $E$-edges. The attachment
target $v_+$ and the far witness $a_\star$ are distinct new elements.
The remaining neighbours are supplied by matching at integer heights.}
\label{fig:worked-certificate}
\end{figure}

Match the remaining rays at displacement $(1,1)$, independently in
each group. Every new element uses its row: copies of $r_A$ have
$E(z,c)$, while copies of $r_\star$ have $E(c,z)$ and choose an
$E$-witness of row $r_V$ in the next group. These choices satisfy
$\beta_5$ for every copy, not just $a_\star$.
All other obligations are already supplied by the interface, the rows,
and the matched rays. Dense completion gives the two underlying orders
without adding successor pairs.

\paragraph*{Necessity of the certificate conditions.}\label{ex:shortcuts}
Splitting a ray with mask $M^+$ into independently chosen neighbours
would violate $S_1=S_2$. Choosing row $r_A$ for $a_\star$ would
conflict with $E(c,a_\star)$, violating~\eqref{eq:pin}.
Agreement on that pair alone does not ensure feasibility: if we replace $\beta_5$ by
\[
 \beta'_5(x,y)=\neg A(x)\lor(E(x,y)\land K(y)),
\]
row $r_A$ remains feasible but $r_\star$ does not. Its source now
needs $E(a_\star,c)$ and cannot use $w$ instead.
The strengthened sentence is unsatisfiable: an $A$-element witnessing
the unique $K$-element's $E$-obligation would require the reverse $E$-edge,
contradicting asymmetry.
\end{document}